\documentclass{article}
\usepackage[utf8]{inputenc}

\usepackage{amsmath,amssymb,amsfonts}
\usepackage{dsfont}
\usepackage{comment}
\usepackage{relsize}

    \usepackage{amsthm, mathtools}
    
	\newtheorem{theorem}{Theorem}

	\newtheorem{lemma}[theorem]{Lemma}
	
	\newtheorem{corollary}[theorem]{Corollary}

	\numberwithin{theorem}{section}
	
	\theoremstyle{definition}
	\newtheorem{question}{Question}

	\def\R{{\mathbb{R}}}
	\def\Z{{\mathbb{Z}}}
	\def\N{{\mathbb{N}}}

	\def\P{{\mathbb{P}}}
	\newcommand{\1}{\mathds{1}}
	
	\def\omg{\omega}
    \def\Omg{\Omega}

	\def\eps{\varepsilon}

	\def\p{\partial}

\def\T{\mathbb{T}}
\def\inter{\mathrm{int}}
\def\exter{\mathrm{ext}}

\renewcommand{\p}{\partial}

\usepackage{xcolor}

\title{Short, Quantitative Construction of the IIC in Planar Percolation}

\author{Malo Hillairet \thanks{Institut Fourier, Université Grenoble Alpes.}}

\date{}

\usepackage{hyperref}
\hypersetup{
    colorlinks=true,
    linkcolor=black,
    citecolor=black,
    urlcolor=black
    }

\begin{document}

\maketitle

\abstract{The classical definitions of the Incipient Infinite Cluster (IIC) of percolation consist in conditioning the origin on being connected to radius $n$ and letting $n$ go to infinity. We provide a short proof of that convergence in the planar setting. A key step in the proof is to introduce an unbiased percolation configuration above which are coupled two percolations conditioned on $0$ being connected to different radii. It implies that the speed of convergence in total variation distance to the IIC measure is upper-bounded by the dual one-arm probability, which is the first occurrence of an explicit upper-bound. Additionally, the proof can be generalised widely to planar graphs. For example, under the assumption that there is no infinite dual cluster at criticality, it proves existence of the IIC on any planar triangulation.}
\section{Introduction}

Bernoulli percolation has been introduced in 1957 by Broadbent and Hammersley \cite{bh57} and refers to models where, in a graph, vertices or edges are independently declared `open' with some probability $p$, and the objects of interest are the macroscopic connectivity properties of the open subgraph. In general percolation models on planar graphs with Euclidean geometry such as bond percolation on $\Z^2$, it is well-established that there is a phase transition from parameters $p < p_c$, at which there are no macroscopic open connected component, to parameters $p > p_c$, at which there is a connected component of positive density. At the critical parameter $p_c \in (0, 1)$, there is no infinite connected component, but there are macroscopic connected components with a rich geometry. In this short article, we provide a new proof of existence of the so-called Incipient Infinite Cluster (IIC) in critical planar Bernoulli percolation models. Although this object had been introduced before, it is Kesten \cite{ke86} who rigorously established its existence on graphs similar to $\Z^2$ as the limit of the sequence of measures
\begin{equation} \label{eqn:convergence}
\lim_{n \rightarrow + \infty} \mathbf{P}_{p_c} \big( \cdot | \, 0 \text{ is connected to distance $n$}\big)
\end{equation}
where $\mathbf{P}_p$ denotes Bernoulli percolation at parameter $p$ and $p_c$ is the critical parameter of the model, and also as
\[
\lim_{\underset{p > p_c}{p \rightarrow p_c}} \mathbf{P}_p \big( \cdot | \, 0 \leftrightarrow \infty\big) \, ,
\]
where $0 \leftrightarrow \infty$ means there is an infinite connected component containing $0$. The IIC thus corresponds to the percolation measure conditioned on $0$ being in an infinite connected component at $p_c$, which is a relevant object in view of understanding better the macroscopic connected components at criticality. There are other equivalent definitions of the IIC, for instance in the works of J\'{a}rai \cite{jar03} and Borgs, Chayes, Kesten, Spencer \cite{bcks01} where the IIC is constructed by considering the largest cluster in the box of size $n$ around the origin and translating it so that it contains $0$. It was also shown by Hammond, Pete and Schramm \cite{hps15} that dynamical percolation observed at a typical exceptional time has the distribution of the IIC (dynamical percolation is the continuous-time process where the states of sites in critical percolation are resampled at exponential rates, and exceptional times are times at which $0$ is connected to infinity by an open path). For background on dynamical percolation and exceptional times, we refer to a survey by Steif \cite{st09}. A natural question left is the speed of convergence in (\ref{eqn:convergence}). In this article, we provide a fairly simple proof of convergence to the IIC which can be applied generally in planar settings and improves on pre-existing upper-bounds on the speed of convergence in some cases (including the triangular lattice).

\subsection{Previously known speeds of convergence} Consider the model of face percolation on the hexagonal lattice, that is, consisting in regular hexagons paving the plane. We independently color the hexagons in black with probability $\frac{1}{2}$ and in white otherwise (as $\frac{1}{2}$ is the critical parameter in this model by Kesten's theorem \cite{ke80}). We denote by $\mathbf{P} = \mathbf{P}_{1/2}$ the associated probability distribution. For $n \in \N$ we write $\Lambda_n$ for the set of hexagons at distance at most $n$ to $0$, and $\partial \Lambda_n$ for its outer boundary $\Lambda_{n+1} \backslash \Lambda_n$. We denote by $\Omg_{\Lambda_n}$ the set $\{0, 1\}^{\Lambda_n}$ of percolation configurations on $\Lambda_n$. We let $\{0 \leftrightarrow \partial \Lambda_n\} \subseteq \Omg_{\Lambda_n}$ be the event that there is a path of black hexagons, in $\Lambda_n$, starting at a neighbor of $0$ and ending at an hexagon adjacent to $\partial \Lambda_n$. We call $\{0 \leftrightarrow \partial \Lambda_n\}$ the \textbf{one-arm event} to radius $n$. Kesten's initial method for constructing the IIC is called the transfer matrix method. Although not explicitly stated in \cite{ke86}, it implies the following speed of convergence in total variation distance, of the sequence of measures $\mathbf{P}(\cdot|_{\Lambda_k} | 0 \leftrightarrow \partial \Lambda_n), n \in \N$.

\begin{theorem}[\cite{ke86}, Transfer-matrix-method speed of convergence] \label{kestenspeed}
For all $\eta \in (0, \frac{1}{2})$, for all $k \leq m \leq n \in \N$ such that $\frac{m}{k}$ is large enough,
\[
\sup_{E \subseteq \Omg_{\Lambda_k}} \Big|\mathbf{P} \big(E \, | \, 0 \leftrightarrow \partial \Lambda_m\big) - \mathbf{P} \big(E \, | \, 0 \leftrightarrow \partial \Lambda_n \big) \Big| \leq \exp \Big(- \log(m / k)^{\frac{1}{2} - \eta} \Big) \, ,
\]
where the supremum is taken over events $E$ depending only on hexagons in $\Lambda_k$.
\end{theorem}

With the same proof, which only relies on `soft' tools of planar percolation such as RSW estimates and the Harris-FKG inequality, this above result extends to similar planar Bernoulli percolation on planar transitive graphs, notably bond percolation on $\Z^2$. The transfer matrix method has been used and extended to construct multiple-arms IIC by Damron and Sapozhnikov \cite{ds11}, IIC in a larger class of graphs such as slabs of $\Z^2$ by Basu and Sapozhnikov \cite{bs17} and recently IIC in higher dimensions by Chatterjee, Chinmay, Hanson and Sosoe \cite{cchs25}. In an unpublished work by Schramm \cite{sch} in the case of the $1$-arm and in \cite[Section~3]{gps13} in the case of the $4$-arm, Schramm and Garban, Pete and Schramm developed a proof of the convergence in (\ref{eqn:convergence}), not using transfer matrices and yielding a better speed of convergence.

\begin{theorem}[\cite{sch, gps13}, Scale-by-scale coupling speed of convergence] \label{gpsspeed}
There exists a positive constant $c$ satisfying that for all $k \leq m \leq n \in \N$ such that $\frac{m}{k}$ is large enough,
\[
\sup_{E \subseteq \Omg_{\Lambda_k}}  \Big| \mathbf{P} \big(E \, | \,  0 \leftrightarrow \partial \Lambda_m\big) - \mathbf{P} \big(E \, | \, 0 \leftrightarrow \partial \Lambda_n \big) \Big| \leq \left(\frac{k}{m}\right)^{c} \, .
\]
\end{theorem}

The constant $c$ in Theorem \ref{gpsspeed} can be thought of as an RSW constant, and it would be tedious to rely on the construction in \cite{gps13} to obtain an explicit lower-bound on $c$. The following question is left open, where `largest' can be understood in the sense of a supremum.

\begin{question} \label{questionc} What is the largest $c$ satisfying the conclusion of Theorem \ref{gpsspeed}, and can it be expressed in terms of arm exponents? 
\end{question}

We briefly expand on the difference with the construction of \cite{ke86} to highlight the progression from the proof of Theorem \ref{kestenspeed} in \cite{ke86} to the proof of Theorem \ref{gpsspeed} in \cite{sch, gps13}, consisting in an `annulus-by-annulus' exploration, to the method of the present article where we use a `site-by-site' exploration. The transfer matrix method of \cite{ke86} relies on considering annuli of width $\exp(i^{\alpha})$, with $\alpha > 1$, and under the assumption that one can find a black circuit surrounding $\Lambda_k$ in \textit{all of the annuli} between $\Lambda_k$ and $\Lambda_m$. The proof of \cite{sch, gps13} uses two additional elements. The first is removing the condition that a black circuit needs to be found in all of the annuli. Instead, we consider annuli of width $\exp(i)$ and it is shown there is a constant probability of achieving coupling at each of these annuli \textit{regardless of what has been revealed in the previous annuli}. The second element is to try achieving coupling by revealing an `innermost black circuit' in each of the annuli, which is more explicit than going through a transfer matrix. The core of the present article is to find another way of coupling, which allows us to reveal an outermost black circuit in $\Lambda_m$. This can be thought of as considering $\Lambda_m \backslash \Lambda_k$ as a very large annulus. As a consequence, we get rid of the need to split that region into nested annuli and obtain a more quantitative speed of convergence.

\subsection{Main result} The main contribution of this article is to provide a short and simple proof of existence of the IIC using standard arguments of percolation theory and statistical mechanics. This partially answers Question \ref{questionc} by deriving a lower-bound on $c$ equal to the one-arm exponent. This is the first time a lower-bound on $c$ is made explicit. To this end, we provide a coupling construction relying on a site-by-site exploration. We explain the proof strategy in more detail soon after stating the result.

\begin{theorem}[Site-by-site coupling speed of convergence] \label{1armbound}
For all $k \leq m \leq n \in \N^*$,
\[
\sup_{E \subseteq \Omg_{\Lambda_k}} \Big|\mathbf{P} \big(E \, | \, 0 \leftrightarrow \partial \Lambda_m\big) - \mathbf{P} \big(E \,| \, 0 \leftrightarrow \partial \Lambda_n \big) \Big| \leq \mathbf{P} \big(\Lambda_k \leftrightarrow^* \partial \Lambda_m\big) \, ,
\]
where $\{\Lambda_k \leftrightarrow^* \Lambda_m\}$ is the dual one-arm event that there is a path of white hexagons from $\Lambda_k$ to $\partial \Lambda_m$.
\end{theorem}

On the triangular lattice $\T$, explicit asymptotics of $\mathbf{P} (\Lambda_k \leftrightarrow^* \partial \Lambda_m)$ are known. Indeed, Lawler, Schramm and Werner \cite{lsw02} proved $\mathbf{P}( \Lambda_k \leftrightarrow \partial \Lambda_m ) = (k / m)^{ 5/48 + o(1)}$ when $\frac{m}{k}$ goes to infinity. The value $\frac{5}{48}$ is called the one-arm exponent. The same exponent is conjectured, but not proved, in critical bond percolation on $\Z^2$. We obtain the following as an immediate corollary.

\begin{corollary} For all $\eta \in (0, \frac{5}{48})$, for all $k \leq m \leq n \in \N$ such that $\frac{m}{k}$ is large enough,
\[
\sup_{E \subseteq \Omg_{\Lambda_k}} \Big|\mathbf{P} \big(E \, | \, 0 \leftrightarrow \partial \Lambda_m\big) - \mathbf{P} \big(E \, | \, 0 \leftrightarrow \partial \Lambda_n \big) \big| \leq \Big( \frac{k}{m} \Big)^{\frac{5}{48} - \eta} \, .
\]
\end{corollary}

\subsection{Generalisation to planar triangulations} Theorem \ref{1armbound} holds with the same proof in any planar model, \textit{regardless of self-duality or symmetry}. To illustrate this, Consider $G = (V, E)$ a planar triangulation \footnote{We say a planar graph is a graph that can be embedded in the Euclidean plane, and we say it is a triangulation if all of its faces are triangles. When working with planar graphs which are not triangulations, the definition of $\{A \leftrightarrow^* B\}$ would need to be adapted, allowing for paths of white vertices to be neighbor through the faces of the graph and not only through the edges. Note that in triangulations, it is equivalent to be neighbor through an edge and through a face. } and let $v_0 \in V$. Denote by $\Lambda_n(v_0)$ the ball of radius $n$ around $v_0$ for the graph distance and $\partial \Lambda_n(v_0) = \Lambda_{n+1}(v_0) \backslash \Lambda_n(v_0)$. We denote by $\mathbf{P}_p^G$ Bernoulli percolation of parameter $p$ on the sites of $G$, and similar definitions as before for connection events. Then, for all integers $k \leq m \leq n$ and any event $E$ depending only on $\Lambda_k(v_0)$ we have
\begin{equation} \label{bound in planar triang}
\Big| \mathbf{P}_p^G \big(E \, |\, v_0 \leftrightarrow \partial \Lambda_m(v_0) \big) - \mathbf{P}_p^G \big(E\,|\, v_0 \leftrightarrow \partial \Lambda_n(v_0) \big) \Big| \leq \mathbf{P}_p^G \big(\Lambda_k(v_0) \leftrightarrow^* \partial \Lambda_m(v_0) \big).
\end{equation}

As a consequence, it is enough that $\mathbf{P}_{p}^G \big(\Lambda_k(v_0) \leftrightarrow^* \partial \Lambda_m(v_0) \big)$ converges to $0$ to define an IIC measure at a parameter $p$ where $\mathbf{P}_p^G(v_0 \leftrightarrow \infty) = 0$. Indeed, consider the measures $\mu_{k, n}$ on $\Omg_{\Lambda_k}$ defined by
\[
\mu_{k,n}(E) = \mathbf{P}_p^G \big(E \, |\, v_0 \leftrightarrow \partial \Lambda_n(v_0) \big)
\]
For every $k$, $(\mu_{k, n})_{n \in \N}$ is a Cauchy sequence in total variation by (\ref{bound in planar triang}), thus convergent. Taking the limit in $k$ follows by Carathéodory's extension theorem. Our proof then has applications in a wide range of planar models where proofs of existence of the IIC are known but are technical and model-dependent. For example, using the above limiting procedure, one can construct a quenched IIC on the UIHPT or other planar maps, and a quenched IIC in Voronoï percolation. We mention that the existence of an annealed IIC on the UIHPT has been proved by Richier \cite{richier18}. Note that Theorem \ref{1armbound} is also valid for bond percolation on planar graphs such as $\Z^2$, as long as the dual connection events are defined in the classical way.

\subsection{Proof overview} The probabilistic tools used in the proof are a coupling based on uniform random variables and a property reminiscent of a spatial Markov property, which already appeared in the works of Kesten and Garban, Pete and Schramm. We immediately give a sketch of proof, which could already be convincing as a full proof. We will couple $\omg^{(m)}$ a percolation conditioned on having a one-arm to distance $m$, and $\omg^{(n)}$ conditioned on having a one-arm to distance $n$. A key ingredient of the proof is to sample both $\omg^{(m)}$ and $\omg^{(n)}$ above the same unconditioned percolation $\omg$, using uniform random variables attached to the hexagons of the lattice. The role of $\omg$ is crucial to define the order in which the uniform random variables are revealed.

In that sketch of proof and in the rest of the paper, given $n \geq 1$ we denote by $\eps$ a generic element of $\Omg_{\Lambda_n}$: for example, for $x \in \Lambda_n$ we abbreviate the set $\{\eps \in \Omg_{\Lambda_n} : \eps_x = 1\}$ by $\{\eps_x = 1\}$, and will use the letter $\eps$ specifically for that notation. The coupled configurations $\omg, \omg^{(m)},\omg^{(n)}$ are related to a distinct, abstract probability space which we define later, and we will call $\mathbb{P}$ the associated probability. When using $\mathbf{P}$, the configurations $\omg, \omg^{(m)}, \omg^{(n)}$ are treated as deterministic.

\begin{proof}[Sketch of proof of Theorem \ref{1armbound}]
Let $k \leq m \leq n$ be fixed integers. We construct a coupling between a percolation configuration $\omg$ having law $\mathbf{P}$, $\omg^{(m)}$ having law $\mathbf{P}(\cdot | 0 \leftrightarrow \partial \Lambda_m)$ and $\omg^{(n)}$ having law $\mathbf{P}(\cdot | 0 \leftrightarrow \partial \Lambda_n)$, through exploring $\omg$ and assigning values to $(\omg_x, \omg^{(m)}_x,\omg^{(n)}_x)_{x \in \Lambda_m}$ site by site using uniform random variables $(U_x)_{x \in \Lambda_m}$. Assume that at some point, the hexagons $x_1, x_2, \dots, x_k$ have been explored, meaning that for each $x \in \{x_1, \dots, x_k\}$, the values of $U_{x}, \omg_x, \omg^{(m)}_x, \omg^{(n)}_x$ are known. From that information, we pick $x_{k+1}$ in $\Lambda_m \backslash \{x_1, \dots, x_k\}$, and set the values of $\omg_{x_{k+1}}, \omg^{(m)}_{x_{k+1}}$ and $\omg^{(n)}_{x_{k+1}}$ according to
\[
\begin{aligned}
\omg_{x_{k+1}} &:= \1 \Big\{U_{x_{k+1}} \leq \frac{1}{2} \Big\} \, , \\
\omg^{(m)}_{x_{k+1}} &:= \1 \Big\{U_{x_{k+1}} \leq \mathbf{P} \big(\eps_{x_{k+1}} = 1 | \,\eps_{x_j} = \omg^{(m)}_{x_j} \, \forall j \leq k, 0 \leftrightarrow \partial \Lambda_m \big) \Big\} \, , \\
\omg^{(n)}_{x_{k+1}} &:= \1 \Big\{U_{x_{k+1}} \leq \mathbf{P} \big(\eps_{x_{k+1}} = 1 | \, \eps_{x_j} = \omg^{(n)}_{x_j} \forall j \leq k, 0 \leftrightarrow \partial \Lambda_n\big) \Big\} \, .
\end{aligned}
\]
The criterion for choosing $x_{k+1}$ is to reveal the outermost black circuit of $\omg$ in $\Lambda_m$ first. Such a circuit exists and separates $\Lambda_k$ from $\partial \Lambda_m$ if and only if there is no white arm from $\Lambda_k$ to $\partial \Lambda_m$ in $\omg$, hence with probability $1 - \mathbf{P}(\Lambda_k \leftrightarrow^*\partial \Lambda_m)$. Using the FKG inequality, we can prove that $\omg^{(m)} \geq \omg$ and $\omg^{(n)} \geq \omg$ almost surely, implying that any black circuit in $\omg$ is also black in $\omg^{(m)}$ and $\omg^{(n)}$. We can see here the importance of $\omg$: we used it to construct an exploration revealing a black circuit common to $\omg^{(m)}$ and $\omg^{(n)}$ while making sure no hexagon has been explored inside that circuit. The final observation is that if such a black circuit $\Gamma$ cutting $\Lambda_k$ off of $\partial \Lambda_m$ has been revealed, then in the region interior to $\Gamma$, conditioning on $\{0 \leftrightarrow \partial \Lambda_m\}$ or on $\{0 \leftrightarrow \partial \Lambda_n\}$ is equivalent to conditioning on $\{0 \leftrightarrow \Gamma \}$. As a consequence and due to their definition, the variables $\omg^{(m)}$ and $\omg^{(n)}$ must coincide in the inner region defined by $\Gamma$, hence in $\Lambda_k$.
\end{proof}

\subsection{Perspectives}
There is no reason for the dual one-arm probability to be a sharp upper-bound in Theorem \ref{1armbound}. In the proof, we are likely to discover more black vertices in $\omg^{(m)}$ and $\omg^{(n)}$ than just the vertices in $\omg$, so maybe the current exploration reveals a black circuit common to $\omg^{(m)}$ and $\omg^{(n)}$ before reaching $\Lambda_k$ with even larger probability. The choice of the exploration is also crucial, and there could be more clever explorations. Thus, Question \ref{questionc} is still open and we expect the answer to be larger than the one-arm exponent $\frac{5}{48}$. We conjecture $c$ is finite, yet there does not seem to be a simple argument to show it, so we leave that problem open as well.

Our proof of Theorem \ref{1armbound} relies on monotonicity of the measures conditioned on having a one-arm with respect to the unbiased measure. Thus, it cannot be adapted as such to arms of different colors, unlike the proofs of \cite{ke86} and \cite{gps13}. When trying the site-by-site exploration proof with arms of different colors, it is no longer easy to find a circuit with respect to which a similar spatial Markov property holds. For monochromatic arms, the proof of Theorem \ref{1armbound} could be adapted as such to the monochromatic two-arm event. It is however unclear what sets could play the role of the black circuits for the monochromatic $k$-arm event with $k \geq 3$. This leads to the question

\begin{question} Can multiple-arms IIC be constructed using a site-by-site exploration, as in the proof of Theorem \ref{1armbound}?
\end{question}

Finally, we briefly comment on larger dimensions. For Bernoulli percolation on $\Z^d$ with $d \geq 3$, there is very low probability of finding a black hypersurface surrouding $0$ at criticality. Thus, it seems there is no hope of adapting the construction of Theorem \ref{1armbound} to larger dimensions. Known methods for constructing the IIC in larger dimensions rely on adaptations of the transfer matrix method, as in \cite{cchs25}, which do not give a polynomial speed of convergence.

\begin{question}
Does convergence to the IIC measure in dimension $d \geq 3$ hold with polynomial speed?
\end{question}

\section{Proofs} \label{proofs}

\subsection{Setting}
 We work with site percolation on the triangular lattice $\T$ having $\Z + (\frac{1}{2}, \frac{\sqrt{3}}{2}) \Z \subseteq \R^2$ as a set of vertices and where edges are between vertices at Euclidean distance $1$ to each other. We represent this model by face percolation on the hexagonal lattice by centering a hexagon of diameter $\sqrt{3}$ on top of each vertex, and colouring in black the hexagons covering open sites and in white the others. We identify sets of vertices $\Lambda \subseteq \Z + (\frac{1}{2}, \frac{\sqrt{3}}{2})\Z$ with subgraphs of $\T$ having full set of edges. We define $\partial \Lambda$, the boundary of $\Lambda$, as the set of vertices at graph distance exactly $1$ to $\Lambda$. For $n \in \N$, we let $\Lambda_n$ be the set of vertices at graph distance to $0$ less than or equal to $n$ in $\T$. Given $\Lambda \subseteq \T$, we let $\Omg_{\Lambda} = \{0, 1\}^{\Lambda}$ be the space of percolation configurations on $\Lambda$. We endow $\Omg_{\T}$ with the Bernoulli percolation measure $\mathbf{P} = \mathrm{Ber}(\frac{1}{2})^{\otimes \T}$. We also introduce an abstract probability space with associated probability $\P$, on which is defined a countable family of independent uniform random variables on $[0, 1]$. The letter $\omg$, with or without superscript, will always be used for percolation configurations defined using these uniform random variables. We reserve the letter $\eta$ to fixed elements of $\Omg_{\Lambda}$, for $\Lambda \subseteq \T$. If $S \subseteq \Lambda$ and $\eta \in \Omg_{\Lambda}$, we denote by $\eta_S$ the restriction of $\eta$ to $S$, and we use similar notation for random percolation configurations. We let $\eps$ denote a generic element of $\Omg_{\Lambda}$: for $S \subseteq \Lambda$ and $\eta \in \Omg_{S}$, we let
\begin{equation} \label{defnofeps}
\{\eps_S = \eta\} := \{\eps \in \Omg_{\Lambda} : \eps_S = \eta\} \, ,
\end{equation}
and we use the letter $\eps$ specifically for that notation. Given $S, S' \subseteq \T$, we write $\{S \leftrightarrow S'\}$ (resp. $\{S \leftrightarrow^* S'\}$) for the event that there is a neighbor-to-neighbor path of vertices $(x_1, \dots, x_k)$ such that $x_1 \in \partial S$, $x_k \in \partial S'$ and $\omg_{x_i} = 1$ (resp. $\omg_{x_i} = 0$) for all $1 \leq i \leq k$. By convention, we declare certain the event $\{S \leftrightarrow S'\}$ if $S$ and $S'$ have non-empty intersection or are adjacent. We will also consider connection events where the connection path is constrained to stay within a region $\Lambda$, which we will denote by $\{S \overset{\Lambda}{\leftrightarrow} S'\}$. Finally, given $n \in \N$, $S \subseteq \Lambda_n$ and $\eta \in \Omg_S$, we say $\eta$ is \textbf{compatible} with $\{0 \leftrightarrow \partial \Lambda_n\}$ if there exists $\eta' \in \Omg_{\Lambda_n}$ such that $\eta'_S = \eta$ and $\eta' \in \{0 \leftrightarrow \partial \Lambda_n\}$.

\subsection{Site-by-site coupling}
We hereby describe what we call an exploration and how we sample percolation configurations conditioned on one-arm events using uniform random variables. Let $\Lambda \subseteq \T$ be a finite set of vertices and $(U_x)_{x \in \Lambda}$ a collection of iid $[0, 1]$-valued uniform random variables. We will denote by $X = (X(i))_{1 \leq i \leq |\Lambda|}$ a random bijection from $\{1, \dots, |\Lambda|\}$ to $\Lambda$, satisfying that for all $1 \leq i \leq |\Lambda| - 1$,
$$
X(i) \text{ is a deterministic function of } \big(X(j)\big)_{1 \leq j \leq i - 1} \text{ and } \big( U_{X(j)} \big)_{1 \leq j \leq i - 1} \, .
$$
We say that $X$ is an \textbf{exploration} of $\Lambda$ associated to $(U_x)_{x \in \Lambda}$. We also let $X_{[i]} := \{X(1), \dots, X(i)\}$ for $0 \leq i \leq |\Lambda|$, this is the set of vertices which have been explored at iteration $i$. Let $n \in \N$ and let $X$ be an exploration of $\Lambda$ associated to $(U_x)_{x \in \Lambda}$. We let $\omg^{(n)}$ be defined by
\begin{equation} \label{omgnfromX}
\omg^{(n)}_{X(i)} := \1 \Big\{U_{X(i)} \leq \mathbf{P}\big(\eps_{X(i)} = 1 |  \eps_{X_{[i-1]}} = \omg^{(n)}_{X_{[i-1]}}, 0 \leftrightarrow \partial \Lambda_n\big) \Big\} \, ,
\end{equation}
where we recall the letter $\eps$ is used for denoting events as in (\ref{defnofeps}) and $\omg_{X_{[i-1]}}, \omg^{(m)}_{X_{[i-1]}},\omg^{(n)}_{X_{[i-1]}}$ and $X(i)$ are treated as deterministic when in input of the probability $\mathbf{P}$. In parallel, we define $\omg$ by
\[
\omg_{X(i)} = \1 \Big\{ U_{X(i)} \leq \frac{1}{2} \Big\} \, .
\]
This is equivalent to setting $\omg_x = \1 \{U_x \leq \frac{1}{2}\}$ for all  $x \in \Lambda$. Hence, $\omg$ has distribution $\mathbf{P}$ and is coupled to $\omg^{(n)}$. The following lemma expresses that defining $\omg^{(n)}$ as in (\ref{omgnfromX}) samples it from the distribution $\mathbf{P}(\cdot | 0 \leftrightarrow \partial \Lambda_n)$, and a simple use of the FKG inequality is sufficient to additionally obtain the monotonicity $\omg^{(n)} \geq \omg$. This lemma and its proof are very similar to \cite[Lemma~2.1]{drt19}.

\begin{lemma} \label{couplingworks}
Let $\Lambda$ be a finite subset of $\T$. Let $n \in \N$ and $(U_x)_{x \in \Lambda}$ be independent $[0, 1]$-valued uniform random variables. Let $X$ be an exploration of $\Lambda$ associated to $(U_x)_{x \in \Lambda}$. Then, the random variable $\omg^{(n)} \in \Omg_{\Lambda}$ defined by (\ref{omgnfromX}) has as a distribution the restriction to $\Lambda$ of $\mathbf{P}(\cdot | 0 \leftrightarrow \partial \Lambda_n)$. Moreover, $\omg^{(n)} \geq \omg$ almost surely.
\end{lemma}

\begin{proof}[Proof of Lemma \ref{couplingworks}.]
Let $n \in \N^*$, $\Lambda \subseteq \T$, $(U_x)_{x \in \Lambda}$, $( X(i))_{1 \leq i \leq |\Lambda|}$ be as in the statement of the lemma, and let $\omg_{\Lambda}, \omg^{(n)}_{\Lambda}$ be defined as in (\ref{omgnfromX}). We prove by induction on $0 \leq i \leq |\Lambda|$: for all $\eta \in \Omg_{\Lambda}$
\[
\P \Big(\omg^{(n)}_{X_{[i]}} = \eta_{X_{[i]}} \Big) = \mathbf{P}\Big(  \eps_{X_{[i]}} = \eta_{X_{[i]}} \big| \, 0 \leftrightarrow \partial \Lambda_n \Big) \, .
\]
Both probabilities are $1$ at $i = 0$. For $i \geq 1$, assume the above holds at $i-1$.  By (\ref{omgnfromX}) and on the event $\{\omg^{(n)}_{X_{[i-1]}} = \eta_{X_{[i-1]}}\}$,
\[
\Big\{ \omg^{(n)}_{X(i)} = 1 \Big\} = \Big\{ U_{X(i)} \leq \mathbf{P} \big( \eps_{X(i)} = 1 | \eps_{X_{[i-1]}} = \eta_{X_{[i-1]}}, 0 \leftrightarrow \partial \Lambda_n \big) \Big\} \, .
\]
Remark that $\omg^{(n)}_{X_{[i-1]}}$ is measurable in $(X(j))_{1 \leq j \leq i-1}$ and the attached uniform random variables, and $X(i)$ is deterministic in this information. Therefore, conditionally on $\{\omg^{(n)}_{X_{[i-1]}} = \eta_{X_{[i-1]}}\}$, $U_{X(i)}$ has a uniform distribution. We deduce
\[
\P \Big( \omg^{(n)}_{X(i)} = \eta_{X(i)} \big| \, \omg^{(n)}_{X_{[i-1]}} = \eta_{X_{[i-1]}} \Big) = \mathbf{P}\Big(\eps_{X(i)} = \eta_{X(i)} \big| \, \eps_{X_{[i-1]}} = \eta_{X_{[i-1]}}, \, 0 \leftrightarrow \partial \Lambda_n \Big) \, .
\]
We conclude using the above and the induction hypothesis:
\begin{align*}
\P \Big(\omg^{(n)}_{X_{[i]}} = \eta_{X_{[i]}} \Big) &= \P \Big( \omg^{(n)}_{X_{[i-1]}} = \eta_{X_{[i-1]}} \Big) \P \Big(\omg^{(n)}_{X(i)} = \eta_{X(i)} \big| \, \omg^{(n)}_{X_{[i-1]}} = \eta_{X_{[i-1]}} \Big)\\
&= \mathbf{P} \Big( \eps_{X_{[i-1]}} = \eta_{X_{[i-1]}} \big| \, 0 \leftrightarrow \partial \Lambda_n \Big) \\
& \quad \quad \quad \quad \times \mathbf{P} \Big(\eps_{X(i)} = \eta_{X(i)} \big| \,  \eps_{X_{[i-1]}} = \eta_{X_{[i-1]}}, \, 0 \leftrightarrow \partial \Lambda_n \Big) \\
&= \mathbf{P} \Big( \eps_{X_{[i]}} = \eta_{X_{[i]}} \big| \, 0 \leftrightarrow \partial \Lambda_n \Big) \, .
\end{align*}
It remains to prove $\omg^{(n)} \geq \omg$. Let $x_1, \dots, x_i \in \Lambda$ be distinct and $\eta \in \Omg_{\Lambda}$. We denote by $x_{[i-1]}$ the set $\{x_1, \dots, x_{i-1}\}$. By the FKG inequality (see e.g. \cite[Theorem~2.4]{gri99}), the measure $\mathbf{P}(\cdot | \eps_{x_{[i-1]}} = \eta_{x_{[i-1]}})$ being a product measure and $\{\eps_{x_i} = 1 \}$, $\{0 \leftrightarrow \partial \Lambda_n\}$ increasing events,
\[
\frac{1}{2} \leq \mathbf{P}\Big(\eps_{x_i} = 1 \big| \, \eps_{x_{[i-1]}} = \eta_{x_{[i-1]}}, 0 \leftrightarrow \partial \Lambda_n \Big) \, .
\]
We deduce that, almost surely,
\[
\1 \Big \{ U_{X(i)} \leq \frac{1}{2} \Big\} \leq \1 \Big\{U_{X(i)} \leq \mathbf{P} \big(\eps_{X(i)} = 1 | \, \eps_{X_{[i-1]}} = \omg^{(n)}_{X_{[i-1]}}, 0 \leftrightarrow \partial \Lambda_n\big) \Big\} \, ,
\] and this completes the proof.
\end{proof}

\subsection{Spatial Markov property with respect to black connected cutsets}

An essential tool for constructing the IIC is a property of independence reminiscent of the spatial Markov property of the percolation measure conditioned on having a one-arm. These measures do not satisfy a general spatial Markov property, however, if we suppose that a connected set which cuts $0$ off of $\partial \Lambda_n$ is black, conditioning on having a one-arm from $0$ to radius $n$ is exactly conditioning on having a one-arm from $0$ to the cutset and a one-arm from the cutset to radius $n$, which are events depending on disjoint sets of hexagons. This observation leads to the conditional independence expressed by (16) in \cite{ke86}, or in the upcoming lemma.

Before stating the lemma, we recall briefly what a cutset is, as these will play a central role. We say that a set $\Gamma \subseteq \T$ is a \textbf{cutset} (of $0$) if and only if $0 \notin \Gamma$ and the connected component of $0$ in $\T \backslash \Gamma$ is finite. If $\Gamma$ is a cutset, we call interior of $\Gamma$ and denote by $\inter(\Gamma)$ the connected component of $0$ in $\T \backslash \Gamma$, and we call exterior of $\Gamma$ and denote by $\exter(\Gamma)$ the set $\T \backslash (\inter(\Gamma) \cup \Gamma)$. Given $\Lambda \subseteq \T$ and a connected cutset $\Gamma$, we say $\Lambda$ is inside $\Gamma$ if $\Lambda \subseteq \inter(\Gamma)$.

\begin{lemma} \label{blackcircuitsmp}
Let $n \geq 0$ and let $\Gamma \subseteq \Lambda_n$ be a connected cutset. Let $S$ be a finite set of vertices containing $\Gamma$ and let $\eta \in \Omg_S$ be such that $\Gamma$ is black in $\eta$ and $\eta$ is compatible with $\{0 \leftrightarrow \partial \Lambda_n\}$. Then, for any event $A$ depending only on coordinates inside $\inter(\Gamma)$,
\[
\mathbf{P} \big( A \, | \, \eps_S = \eta, 0 \leftrightarrow \partial \Lambda_n \big) = \mathbf{P} \big( A \, | \, \eps_S = \eta, 0 \leftrightarrow \Gamma \big) = \mathbf{P} \big( A \, | \, \eps_{S \cap \inter(\Gamma)} = \eta_{S \cap \inter(\Gamma)}, 0 \leftrightarrow \Gamma \big) \, .
\]
\end{lemma}

\begin{proof}[Proof of Lemma \ref{blackcircuitsmp}.]
Let $n \in \N$, $\Gamma \subseteq S \subseteq \Lambda_n$ with $\Gamma$ a connected cutset and let $\eta \in \Omg_S$ be as in the statement of the lemma. First observe
\[
\big\{\eps_S = \eta, 0 \leftrightarrow \partial \Lambda_n \big\} = \big\{\eps_S = \eta, 0 \overset{\inter(\Gamma)}{\leftrightarrow} \Gamma, \Gamma \overset{\exter(\Gamma)}{\leftrightarrow} \partial \Lambda_n \big\} \, .
\]
Indeed, since $\Gamma$ is a cutset of $0$ inside $\Lambda_n$, any black one-arm from $0$ to $\partial \Lambda_n$ connects $0$ to $\Gamma$ and $\Gamma$ to $\partial \Lambda_n$. Conversely, since $\Gamma$ is black in $\eta$ and connected, if $0$ is connected to $\Gamma$ by a black path $\pi_{\inter}$ and $\Gamma$ is connected to $\partial \Lambda_n$ by a black path $\pi_{\exter}$, the union of sets $\pi_{\inter} \cup \Gamma \cup \pi_{\exter}$ is black and connects $0$ to $\partial \Lambda_n$. Now, let $A$ be an event depending only on $\inter(\Gamma)$. By independence of percolation events related to disjoint sets,
\[
\mathbf{P} \big(A \cap \{0 \overset{\inter(\Gamma)}{\leftrightarrow} \Gamma\} \cap \{\Gamma \overset{\exter(\Gamma)}{\leftrightarrow} \partial \Lambda_n\}|\eps_S = \eta \big) = \mathbf{P} \big(A \cap \{0 \overset{\inter(\Gamma)}{\leftrightarrow} \Gamma\}\big| \, \eps_S = \eta \big)\mathbf{P} \big(\Gamma \overset{\exter(\Gamma)}{\leftrightarrow} \partial \Lambda_n \big| \,\eps_S = \eta \big) \, .
\]
By the definition of conditional expectation and applying the above to the certain event $\Omg$ as well, we deduce
\[
\mathbf{P} \big(A \, | \, \eps_S = \eta, 0 \leftrightarrow \partial \Lambda_n \big) = \mathbf{P} \big( A \, | \, \eps_S = \eta, 0 \leftrightarrow \Gamma \big) \, .
\]
Finally, since the events $A$ and $\{0 \leftrightarrow \Gamma\}$ depend only on $\eps_{\inter(\Gamma)}$, we have
\[
\mathbf{P} \big( A \cap \{0 \leftrightarrow \Gamma\} \, | \, \eps_S = \eta_S \big) = \mathbf{P} \big( A \cap \{0 \leftrightarrow \Gamma\}\, | \, \eps_{S \cap \inter(\Gamma)} = \eta_{S \cap \inter(\Gamma)} \big) \, ,
\]
which amounts to the second equality of the lemma after dividing by
\[
\mathbf{P} \big(0 \leftrightarrow \Gamma \, | \, \eps_S = \eta_S \big) = \mathbf{P} \big( 0 \leftrightarrow \Gamma \, | \, \eps_{S \cap \inter(\Gamma)} = \eta_{S \cap \inter(\Gamma)} \big) \, . \qedhere
\]
\end{proof}

\subsection{The boundary of a simply connected set is connected}

Finally, we will need that in $\T$, the boundary of a simply connected set is connected. This would hold as well in any planar triangulation. For general planar graphs, one would need to define the boundary differently: Lemma \ref{boundary of connected is connected} does not hold as such with sites of $\Z^2$ for example.

\begin{lemma} \label{boundary of connected is connected}
Let $\Lambda \subseteq \T$ be a finite, connected set such that $\T \backslash \Lambda$ is connected. Then the set $\p \Lambda$ is connected.
\end{lemma}

In order to prove the above lemma, we use the following version of a fundamental percolation fact, that a topological rectangle coloured in white and black is crossed by a black path from left to right if and only if it is not crossed by a white path from top to bottom. We only state that the absence of a white crossing implies the existence of a black crossing in the other direction. Indeed, this implication is sufficient to prove Lemma~\ref{boundary of connected is connected}, and we cannot state an equivalence unless we enter into more involved topological considerations.

\begin{lemma} \label{fundamental perco fact}
Let $\Gamma_1, \Gamma_2, \Gamma_3, \Gamma_4$ be four self-avoiding, disjoint paths of $\T$ such that the concatenation $\Gamma = \Gamma_1 \sqcup \Gamma_2 \sqcup \Gamma_3 \sqcup \Gamma_4$ is a loop. For any black-and-white colouring of $\T$ such that $\Gamma_1$ and $\Gamma_3$ are black, if there is no white connected component intersecting both $\Gamma_2$ and $\Gamma_4$, then $\Gamma_1$ and $\Gamma_3$ are in the same connected component.
\end{lemma}

We refer to \cite{bolrio06}, (see Chapter~1 and Lemma~5 in Chapter~7) for a more complete exposition about fundamental planar percolation facts, and proofs of results similar to our Lemma~\ref{fundamental perco fact}.

\begin{proof}[Proof of Lemma~\ref{boundary of connected is connected}]
Let $\Lambda \subseteq \T$ be a connected set such that $\T \backslash \Lambda$ is connected. Consider $x, y \in \partial \Lambda$. We prove $x$ and $y$ are connected in $\p \Lambda$. Since $\Lambda$ is connected, there is a path $\Gamma_2$ from a neighbour of $x$ to a neighbour of $y$ in $\Lambda$, and since $\T \backslash \Lambda$ is connected  there is a path $\Gamma_4$ from a neighbour of $y$ to a neighbour of $x$ in $\T \backslash \Lambda$. By loop-erasure, we can assume $\Gamma_2$ and $\Gamma_4$ are self-avoiding. The paths $\Gamma_1 = (x),  \Gamma_2, \Gamma_3 = (y), \Gamma_4$ then satisfy the assumptions of Lemma~\ref{fundamental perco fact}. For each $v \in \T$, colour $v$ in black if $v \in \p \Lambda$ and white otherwise. A path from a vertex of $\Lambda$ to a vertex of $\T \backslash \Lambda$ has a first vertex in $\T \backslash \Lambda$, which must be black. Hence, there is no white connected component intersecting both $\Gamma_2$ and $\Gamma_4$. By Lemma~\ref{fundamental perco fact}, we deduce $x$ and $y$ are in the same black connected component, which implies $x$ and $y$ are connected in $\p \Lambda$.
\end{proof}

\subsection{Proof of Theorem \ref{1armbound}}

We conclude by proving Theorem \ref{1armbound}. It is a formalized version of the sketch of proof given in the introduction, making use of Lemmas \ref{couplingworks}, \ref{blackcircuitsmp} and \ref{boundary of connected is connected}.

\begin{proof}[Proof of Theorem \ref{1armbound}.]
Let $k \leq m \leq n$ and let $(U_x)_{x \in \Lambda_m}$ be iid uniform random variables on $[0, 1]$. From these variables we construct $\omg$, $\omg^{(m)}$ and $\omg^{(n)}$: three percolation configurations in $\Lambda_m$ satisfying that $\omg$ has law $\mathbf{P}$, $\omg^{(m)}$ has law $\mathbf{P}(\cdot | 0 \leftrightarrow \partial \Lambda_m)$, $\omg^{(n)}$ has law $\mathbf{P}(\cdot | 0 \leftrightarrow \partial \Lambda_n)$ and $\omg^{(m)} \geq \omg$ and $\omg^{(n)} \geq \omg$. According to Lemma \ref{couplingworks}, it is sufficient for these properties to hold to introduce an exploration $\{ X(i), 1 \leq i \leq |\Lambda_m| \}$ of the vertices in $\Lambda_m$, associated to $(U_x)_{x \in \Lambda_m}$, then define $\omg$, $\omg^{(m)}$, $\omg^{(n)}$ according to (\ref{omgnfromX}). Additionally, we define $\{X(i), 1 \leq i \leq |\Lambda_m|\}$ carefully so that it reveals the outermost black circuit of $\omg$ before revealing anything in its interior, as follows. For $1 \leq i \leq |\Lambda_m|$, let $E_{i-1}$ be the set of sites of $\Lambda_m \backslash X_{[i-1]}$  connected by a white path to $\partial \Lambda_m$ in $\omg_{X_{[i-1]}}$. We choose the next site to reveal $X(i)$ by picking sites in $E_{i-1}$ in priority: if $E_{i-1}$ is non-empty, pick $X(i)$ arbitrarily in $E_{i-1}$, otherwise pick $X(i)$ arbitrarily in the set of unrevealed sites $\Lambda_m \backslash X_{[i-1]}$ (here `arbitrarily in $E_{i-1}$' means any choice of $X(i) \in E_{i-1}$ that is deterministic in $X_{[i-1]}$ and $(U_{X(j)})_{j \leq i - 1}$ is satisfactory). Note that $E_{i-1}$ also depends only on $X_{[i-1]}$ and $(U_{X(j)})_{j \leq i-1}$ so that this is well-defined and $X$ then satisfies the conditions of Lemma \ref{couplingworks}.

Our goal is to prove that $\omg^{(m)}$ and $\omg^{(n)}$ agree in $\Lambda_k$ as soon as $\omg \notin \{ \Lambda_k \leftrightarrow^* \partial \Lambda_m\}$. Define the stopping-time $\tau$ by
\[
\tau := \min \big\{ 0 \leq i \leq |\Lambda_m| : E_i = \emptyset \big\} \, . 
\]
Hence, $\tau$ is the first time in the exploration at which no unrevealed site is connected to $\partial \Lambda_m$ by a white path in $\omg_{X[\tau]}$. From now on, assume the event $\{X_{[\tau]} \cap \Lambda_k = \emptyset\}$. This event is exactly the event that the white connected component of $\partial \Lambda_m$ in $\omg$ and its black boundary have been revealed before the exploration has reached $\Lambda_k$, so that
\begin{equation} \label{equalityofevents}
\big\{X_{[\tau]}  \cap \Lambda_k \neq \emptyset \big\} = \big\{\Lambda_k \leftrightarrow^* \partial \Lambda_m \text{ in } \omg\big\} \, .
\end{equation}
Let $C$ be the connected component of $0$ in $\Lambda_m \backslash X_{[\tau]}$. By convention, a site neighbor to $\partial \Lambda_m$ is always said to be connected by a white path to $\partial \Lambda_m$, so that every vertex in $\p \Lambda_{m-1}$ is revealed before $\tau$. Hence $C \cap \Lambda_{m-1}$ is empty and therefore $\p C$ is a cutset and is inside $\Lambda_m$. The set $C$ is connected by definition as a connected component. Moreover, $\T \backslash C$ is connected because any vertex in $\T \backslash C$ is connected to $\partial \Lambda_m$. By Lemma \ref{boundary of connected is connected}, we deduce $\p C$ is connected. By maximality of $C$, every $x \in \p C$ is in $X_{[\tau]}$, and $C = \inter(\p C)$ by definition of the interior of a cutset. Since we assumed $\{X_{[\tau]} \cap \Lambda_k = \emptyset\}$, we have that $\p C$ is around $\Lambda_k$. We thus have proved $\p C$ is a connected cutset and $\Lambda_k$ is inside $\p C$. Now remark that $\p C$ is black in $\omg_{X_{[\tau]}}$ by minimality of $\tau$. By the second part of Lemma \ref{couplingworks}, $\omg^{(m)} \geq \omg$ and $\omg^{(n)} \geq \omg$ so that every site in $\p C$ is black as well in $\omg^{(m)}$ and in $\omg^{(n)}$.

We now prove by induction on $i \geq \tau$ that for all $x \in X_{[i]} \cap C$, $\omg^{(m)}_x= \omg^{(n)}_x$. The set $X_{[i]} \cap C$ is indeed empty at $i = \tau$, and for $i \geq \tau + 1$, by the first equality in Lemma~\ref{blackcircuitsmp} applied to $\Gamma = \p C$,
\[
\mathbf{P} \Big( \eps_{X(i)} =1 \big| \, \eps_{X_{[i-1]}} = \omg_{X_{[i-1]}}^{(m)}, 0 \leftrightarrow \partial \Lambda_m \Big) = \mathbf{P} \Big(\eps_{X(i)} = 1 \big| \, \eps_{X_{[i-1]}} = \omg_{X_{[i-1]}}^{(m)}, 0 \leftrightarrow \Gamma \Big) \, ,
\]
keeping in mind that $X$, $\omg^{(m)}$ and $\p C$ are treated as deterministic when using $\mathbf{P}$, although these are random under $\P$. By the second equality in Lemma~\ref{blackcircuitsmp}, we can replace the conditioning on $\{\eps_{X_{[i-1]}} = \omg_{X_{[i-1]}}^{(m)}\}$ by a conditioning on $\{\eps_{X_{[i-1]} \cap C} = \omg_{X_{[i-1]} \cap C}^{(m)}\}$. Under the induction hypothesis that $\omg^{(n)} = \omg^{(m)}$ on $X_{[i-1]} \cap C$, we deduce
\[
\mathbf{P} \Big( \eps_{X(i)} =1 \big| \,  \eps_{X_{[i-1]}} = \omg_{X_{[i-1]}}^{(m)}, 0 \leftrightarrow \partial \Lambda_m \Big) = \mathbf{P} \Big(\eps_{X(i)} = 1 \big| \, \eps_{X_{[i-1]}} = \omg_{X_{[i-1]}}^{(n)}, 0 \leftrightarrow \partial \Lambda_n \Big) \, .
\]
By (\ref{omgnfromX}), we deduce that if $\omg^{(n)} = \omg^{(m)}$ on $X_{[i-1]} \cap C$ and if $X(i) \in C$ then $\omg^{(n)}_{X(i)} = \omg^{(m)}_{X(i)}$, which concludes the induction. Therefore, on the event $\{X_{[\tau]} \cap \Lambda_k = \emptyset\}$, $\omg^{(m)}$ and $\omg^{(n)}$ coincide in $\Lambda_k$, so that
\[
\big\{\omg_{\Lambda_k}^{(m)} \neq \omg_{\Lambda_k}^{(n)} \big\} \subseteq \big\{ X_{[\tau]}  \cap \Lambda_k \neq \emptyset \big\} \overset{(\ref{equalityofevents})}{=} \big\{ \Lambda_k \leftrightarrow^* \partial \Lambda_n  \big\}\, .
\]
Finally, if $E$ is an event depending only on $\Lambda_k$,
\[
\Big|\mathbf{P} \big(E \, | \, 0 \leftrightarrow \partial \Lambda_m \big) - \mathbf{P} \big(E \, | \, 0 \leftrightarrow \partial \Lambda_n \big) \Big| \leq \P \big(\omg^{(m)}_{\Lambda_k} \neq \omg^{(n)}_{\Lambda_k} \big)  \leq \P \big(\Lambda_k \leftrightarrow^* \partial \Lambda_n \big) \, . \qedhere
\]
\end{proof}

\section{Acknowledgments}
The author would like to thank Hugo Vanneuville and Vincent Beffara for instructive discussions and detailed comments on the manuscript, and Daniel de la Riva and Christophe Garban for private communications and remarks at an early stage of the project. The author would like to acknowledge Christophe Garban for communicating the unpublished notes by Schramm \cite{sch}.

\bibliographystyle{amsalpha}
\bibliography{iicbiblio}

\noindent
Malo Hillairet:\\
Institut Fourier, UMR 5582, Laboratoire de Mathématiques, Université Grenoble Alpes, CS 40700, 38058 Grenoble cedex 9, France\\
Email: \href{mailto:malo.hillairet@univ-grenoble-alpes.fr}{malo.hillairet@univ-grenoble-alpes.fr} \\
Url: \href{https://www-fourier.univ-grenoble-alpes.fr/~hillairm/}{https://www-fourier.univ-grenoble-alpes.fr/~hillairm/}

\end{document}